\documentclass[11pt,a4paper]{article}
\usepackage{amssymb,amsmath,amsthm,amstext}
\usepackage{flushend}
\usepackage{authblk}
\usepackage{palatino,a4}

\usepackage{tikz}
\usetikzlibrary{intersections,positioning,shapes,calc,decorations.pathmorphing,decorations.pathreplacing}

\definecolor{myblue}{RGB}{80,80,160}
\definecolor{mygreen}{RGB}{80,160,80}

\usepackage{color}
\renewcommand{\marginpar}[1]{}

\if01
\fi

\newtheorem{theorem}{Theorem}[section]
\newtheorem{lemma}[theorem]{Lemma}

\newtheorem{claim}[theorem]{Claim}
\newtheorem{proposition}[theorem]{Proposition}
\newtheorem{definition}[theorem]{Definition}
\newcommand{\zo}{\{0,1\}}
\newcommand{\ie}{i.e.}
\newcommand{\mapping}{\rightarrow}
\newcommand{\prob}{\rm Prob}
\newcommand{\poly}{\text{poly}}

\usepackage{epsfig}
\usepackage{amsmath}

\usepackage{amsfonts}

\begin{document}

\title{Linear list-approximation  for  short programs (or the power of a few random bits)}
\author[1]{Bruno Bauwens}
\affil[1]{Universit\'e de Lorraine}
\affil[2]{Towson University}
\author[2]{ Marius Zimand\footnote{
  A preliminary version has been presented at the 29-th IEEE Conference on Computational Complexity, June 11-13, 2014, Vancouver, British Columbia, Canada.
  \vspace{5pt} 
  \\Bruno Bauwens is grateful to Mathieu Hoyrup and Universit\'e de Lorraine for financial
  support. Marius Zimand was partially supported by NSF grant CCF 1016158.
  \vspace{5pt} 
  \\Author's addresses: Bruno Bauwens; 
  Department of Applied Mathematics, Computer Science and Statistics, Ghent University, 
  Krijgslaan 281, S9,  9000 Gent, Belgium; http://www.bcomp.be;  and Marius Zimand, 
  Department of Computer and Information Sciences, Towson University,
  Baltimore, MD.; 
  http://triton.towson.edu/\~{ }mzimand.
  }
  }
\maketitle

\begin{abstract}
A $c$-short program for a string $x$ is a description of $x$ of length at most $C(x) + c$, where $C(x)$ is the Kolmogorov complexity of $x$. We show that there exists a randomized algorithm that constructs a list of $n$ elements that contains a $O(\log n)$-short program for $x$. We also show a polynomial-time randomized construction that achieves the same list size for $O(\log^2 n)$-short programs. These results beat the lower bounds shown by Bauwens et al.~\cite{bmvz:c:shortlist} for deterministic constructions of such lists.  We also prove tight lower bounds for the main parameters of our result. The  constructions use only $O(\log n)$  ($O(\log^2 n)$ for the polynomial-time result) random bits. Thus using only few random bits it is possible to do tasks that cannot be done by any deterministic algorithm regardless of its running time.
\end{abstract}





%
%

\section{Introduction}
The Kolmogorov complexity of a string $x$, denoted $C(x)$, is the length of a shortest description of $x$ relative to a fixed universal Turing machine $U$. In many applications, it is desirable to represent information $x$ in a succinct form, \ie, to find a string $p$ such that $U(p) = x$ (such a $p$ is called a \emph{program} for $x$) with length $|p| \approx C(x)$.  Unfortunately this is not possible: Not only is $C(x)$ uncomputable, but it cannot be even approximated in a useful way. Indeed, while the upper bound $C(x) < |x| + O(1)$ is immediate, Zvonkin and Levin~\cite{zvo-lev:j:kol} have shown that no unbounded computable function can lower bound $C(x)$. Beigel et al. ~\cite{bei-bur:j:enumerations} have investigated the \emph{list-approximability} of $C(x)$, \ie, the possibility of constructing a short list of numbers guaranteed to contain $C(x)$. They show that there exists a constant $a$ (which depends on the universal machine $U$) such that any computable list containing $C(x)$ has size $n/a$ (where $n$ is the length of $x$). Since it is trivial to obtain a list of size $n+ O(1)$ that contains $C(x)$, the result of~\cite{bei-bur:j:enumerations} implies that no list-approximation is possible with
  lists  significantly  shorter than the trivial one.

 In view of these strongly negative facts, the recent results of Bauwens, Mahklin, Vereshchagin and
 Zimand~\cite{bmvz:c:shortlist} and Teutsch~\cite{teu:j:shortlist} are surprising. They show that it
 is possible to effectively construct a short list guaranteed to contain a close-to-optimal program
 for $x$.  Even more, in fact the short  list can be computed in polynomial time. More
 precisely,~\cite{bmvz:c:shortlist}  showed that one can compute lists of quadratic 
 size guaranteed to contain a program of $x$ whose length  is $C(x) + O(1)$ 
 and  that one can compute in polynomial-time a list guaranteed to contain a program whose length is
 additively within  $C(x) + O(\log n)$. 
 \cite{teu:j:shortlist}  improved the latter result  by
 reducing the $O(\log n)$ term to  $O(1)$ (see also~\cite{zim:c:shortlistshortproof} for a simpler proof).

In this paper, we investigate how short a computable list that contains a succinct program for $x$ can be.  The size of the  list in~\cite{bmvz:c:shortlist} is \emph{quadratic} in $n$ and in fact in the same paper it is shown that this is optimal because any effectively computed list that contains a program that is additively $c$ close to optimal length must have size $\Omega(n^2/(c+1)^2)$ (for any $c$). The size of the list in the polynomial-time construction from~\cite{teu:j:shortlist} is $n^{7 + \epsilon}$ and~\cite{zim:c:shortlistshortproof} improves it to $O(n^{6 + \epsilon})$. 
We show here that the size of the list can be \emph{linear},  thus beating the above quadratic lower bound,   if we allow \emph{probabilistic computation}, in fact even \emph{polynomial-time probabilistic computation}. Namely, we show  that there exists a  probabilistic algorithm that on input $x$ of length $n$ produces a list of $n$ elements, that, with high probability, contains a  program of $x$ which is additively within $O(\log  n)$ from optimal. We also show the existence of a polynomial-time algorithm with the same property but for $O(\log^2 n)$ closeness to optimality.  The lower bound mentioned above shows that such a list cannot be deterministically computed regardless of the running time. Furthermore, the first algorithm uses only  $O(\log n)$ random bits, and the polynomial-time algorithm uses only $O(\log^2 n)$ random bits. The relevance of these facts will be discussed shortly. These results are shown in Section~\ref{s:ub}. In Section~\ref{s:lb}, we prove tight lower bounds which show that our results are essentially optimal. More precisely, we consider the parameters $c, T, r$ in our main result which  are defined as follows: (1) $c$ is the closeness to $C(x)$ of the length of the desired succinct program, (2) $T$ is the size of the list guaranteed to contain such a succinct program, (3) $r$ is the number of random bits used in the probabilistic construction of the list. In the main result we obtain $c = O(\log n)$, $T = n$, and $r = O(\log n)$. We show that essentially none of these parameters can be improved while keeping the other two the same.  
\smallskip

\textbf{Discussion: The power of randomized computation.}
Can we solve using randomness tasks that cannot be solved without? This is a foundational question,
and its exploration has an old history~\cite{lee-moo-sha-sha:j:probcomput,zvo-lev:j:kol} and a
recent history~\cite{bie-pat:t:probcomput,bie:c:prob,rum-shen:t:probconstruct}.

  There are fields where randomness plays an essential role as a conceptual tool,  \ie, as an element introduced in the model. Pre-eminent examples are Game Theory (the utilization of mixed strategy) and Cryptography (the utilization of secret keys). 
The answer to the above general question is less clear if we restrict to computational tasks. There is a common perception that randomized computation is not fundamentally more powerful than deterministic computation, in the sense that whenever a randomized process solves a task, there also exists a deterministic solution (albeit, often, a slower one). This perception is caused by the simple observation that a probabilistic algorithm can be simulated deterministically after which one can take a majority vote. A similar argument  works for tasks computing an infinite object and the classical theorem of de Leeuw, Moore, Shannon and Shapiro~\cite{lee-moo-sha-sha:j:probcomput} states that if a function can be computed by a probabilistic algorithm, then it can be computed deterministically. However, these considerations only apply to tasks admitting a unique solution. For tasks admitting multiple solutions, randomness could potentially  be helpful.

The task of computing a string with high Kolmogorov complexity is usually given as an example to
illustrate the power of randomized computation (see for
example~\cite{zvo-lev:j:kol,rum-shen:t:probconstruct}): The task cannot be solved deterministically,
but an algorithm that tosses a coin  does it easily by just printing  the coin flips.  However this
example is trivial and not very convincing because the noncomputable output of the above procedure
is exactly the noncomputable part introduced in the procedure. More precisely, if  $f$ is the
probabilistic algorithm with input $x$ and random bits $r$, then $f(x,r) = r$.  Let us consider
another task: On input $x$, find an extension of it called $y$ such that $y$ has larger complexity
than $x$. This second task can be solved in the same trivial way by obtaining via coin tosses a
string $r$ and then taking $y=xr$. Note that this time we can have $|r| \ll |y|$.  
For infinite objects, better examples are known.  N.V.~Petri 
(see~\cite{rum-shen:t:probconstruct}) showed that with positive probability, one can enumerate a
graph of a total function $f$ that exceeds all computable functions. The procedure utilizes a
polynomial number of random bits to generate $f(x)$. Obviously, the time to generate
$f(x)$ from $x$ is not bounded by any computable function. The reader can find other similar
examples in~\cite{bie-pat:t:probcomput}.

So the interesting question is whether 
there are \emph{non-trivial} computational tasks involving finite objects that can be solved probabilistically (perhaps even in polynomial time)  but not deterministically. Furthermore,  if the answer is positive,  can the amount of random bits necessary to solve such a task be very low? In other words, is it the case that even very few random bits can solve a non-trivial task, which is deterministically unsolvable? Our main results give positive answers to  these questions.

\if01
Any definition of \emph{trivial task} (and, by opposition, of \emph{non-trivial task}) is inherently debatable. We sketch here one attempt. Intuitively, a task is trivial if a solution can be obtained by a simple combination of the input and a random string, as in the above examples.  A task is defined by a binary predicate $P(x,y)$. We view $x$ as the input instance and in case $y$ is such that $P(x,y)=1$, we say that $y$ is a solution for $x$. The task is: On input $x$, find a solution $y$ for $x$. We also assume that tasks are nice enough that solutions exist only among strings of bounded length in the length of the input. Thus, there exists a computable function $b$ such that if $P(x,y) = 1$ then $|y| \leq b(|x|)$.    The task is trivial if there exists a \emph{simple} function  $g(x,r)$  such that $|g(x,r)| \leq b(|x|)$ for all $r$ of a given length (which depends on $|x|$)  and for  $99$\% of $r's$ it holds that $g(x,r)$ is a solution for $x$.  By \emph{simple function}, we mean a function that is a composition of a projection and a permutation (or we can take a more general  stance and say that a function is simple if it is computable by a uniform boolean or arithmetic ${\rm NC}^0$ circuit).
\fi

Any definition of \emph{trivial task} is inherently debatable. For our discussion it is sufficient to use an informal formulation whose requirements are so minimal that it is unlikely to raise controversy. 
Intuitively, a task is trivial if a solution for it  can be ``read" almost directly from the pair consisting of the input and a random string. For concreteness, we interpret ``read" as the composition of a projection and a permutation, or  we can take a more general  stance and interpret it as the computation of a uniform boolean or arithmetic ${\rm NC}^0$ circuit. 

Now consider the following task (which depends on a parameter $c$): Given $x$ of length $n$, find $y$, a list of $n$ strings such that at least one of them is a  program for $x$ of length bounded by $C(x) +c$.

 For $c = o(\sqrt{n})$,  by the lower bound from~\cite{bmvz:c:shortlist},  we know that it cannot be solved by deterministic algorithms. The results in this paper show that, for $c = O(\log n)$,  the task  is  solvable by a randomized algorithm,  which remarkably uses only $O(\log n)$ random bits. Furthermore, for $c = O(\log^2 n)$, the task is solvable by a polynomial-time algorithm that uses only $O(\log^2 n)$ random bits. The task appears to be non-trivial, at least we are not aware of any ${\rm NC}^0$ solution. 

In Section~\ref{s:randdeterm}, we elaborate the above example and  present an even more natural
non-trivial task that (1) can be solved in polynomial time by a randomized algorithm, and (2) cannot
be solved by any deterministic algorithm that runs in computably bounded time. 
This task asks that on every input $(x, \ell)$,  if $\ell = C(x)$, a string $z$ should be constructed such that  $z$ is a short program for~$x$. This is a promise problem, because in case $\ell \not= C(x)$,  the algorithm is not even required to  halt, or, in case it halts, $z$ can be an arbitrary string. Note that on input $(x, \ell)$, if $\ell = C(x)$, one can simply by exhaustive search find a shortest program for $x$. However, the exhaustive search does not halt if $C(x) > \ell$. Actually, we show that there is no deterministic algorithm that, in case $\ell = C(x)$,  
runs in computably bounded time and constructs a $o(n)$-short program for $x$.  On the other hand,
relying on the techniques used in the proof of our main results, we present a randomized algorithm that runs  in polynomial time  and constructs   with probability $(1-\delta)$  a $O(\log^2 (n/\delta))$-short program for $x$,  conditioned that the promise $\ell = C(x)$ holds.  Furthermore, the randomized algorithm uses only  $O(\log^2 n/\delta)$ random bits. In fact, using relativized versions of this task, we notice that polynomial-time randomized computation can be more powerful than deterministic computation that lies arbitrarily high  in the arithmetic hierarchy.

It remains to investigate if there exist non-trivial tasks that cannot be solved deterministically
but that can be solved using even fewer random bits (e.g., $o(\log n)$ or even $O(1)$ random bits).
In the Appendix, we give an example that can be solved with $O(1)$ random bits, but it is still
borderline trivial, because the solution is obtained simply by dividing the input length by a
constant. 

\section{Preliminaries}

We fix a universal Turing machine $U$ that is \emph{standard}  (meaning that for every machine $V$ there is a polynomial-time computable function $t$ such that,  for all $p$,  $U(t(p)) = V(p)$ and  $|t(p)| = |p| +O(1)$.)

 $C$ stands for the plain Kolmogorov complexity relative to  $U$. Thus, for any string $x$, $C(x) = \min \{|p| \mid U(p)=x\}$.
If $U(p) = x$, we say that $p$ is a program for~$x$. If in addition, $|p| \leq C(x) + c$, then we say that $p$ is a $c$-short program for $x$.

We use  bipartite graphs $G = (L,R,E \subseteq L \times R)$ with $L = \zo^n$ (or, a few times, $L \subset \zo^n$),  $R= \zo^m$ and which are left-regular, \ie, all the  nodes  in $L$ have the same degree, which we denote $2^d$. We denote $N=2^n, M = 2^m, D = 2^d$. 



As it is typically the case, we actually work with a family of graphs indexed on~$n$ and such a
family of graphs is \emph{constructible}  if there is an algorithm that on input $(x,y)$, where $x \in \zo^n = L$ and $y \in \zo^d$,  outputs  the $y$-th neighbor of $x$. Some of the graphs also depend on a rational $0 < \delta < 1$.  
A constructible family of  graphs is \emph{explicit} if the above algorithm runs in time $\poly(n, 1/\delta)$.

A $(k, \epsilon)$ extractor is a function $E: \zo^n \times \zo^d \mapping  \zo^m$ such that for any distribution $X$ on $\zo^n$ with min-entropy $H_\infty(X) \geq  k$, $E(X,U_d)$ is $\epsilon$-close to $U_m$, where $U_d$ ($U_m$) is the uniform distribution on $\zo^d$ (respectively, $\zo^m$), \ie, for every $A \subseteq R$,
\begin{equation}
 \label{eq:extractorDef}
\bigg| \prob[E(X, U_d) \in A] - \frac{|A|}{M} \bigg| < \epsilon.
\end{equation}
 It is known that it is enough to require that the condition holds for all distributions $X$ that are ~\emph{flat} \cite{cho-gol:j:weaksource}. The value $k+d-m$ is called the entropy loss of the extractor.
We remind the reader that $H_\infty(X) \geq k$ means that for any $x \in \zo^n$, $\prob_X(x) \leq 2^{-k}$ and that a distribution is flat if it assigns equal probability mass to each element in its support.

To an extractor $E$, we associate the bipartite graph $G_E$ with $L = \zo^n, R= \zo^m$, such that for every $x \in L$, its neighbors are $N(x) = \{E(x,y) \mid y \in \zo^d\}$.

In this paper we need extractors for which $k+d - m =   O(\log (n/\epsilon))$, \ie, the entropy loss is at most logarithmic.
Using standard probabilistic methods the following extractor  can be shown  to exist, which has even smaller entropy loss (see~\cite{rad-tas:j:extractors}).

\begin{theorem}
\label{t:extopt}
For all $n$, $k \leq n$, and $\epsilon > 0$, there exists a $(k,\epsilon)$ extractor with $m = k + d - 2 \log (1/\epsilon) - O(1)$ and $d = \log(n-k) + 2 \log (1/\epsilon) + O(1)$.
\end{theorem}

The above result is existential, but using the fact that the number of flat distributions with min-entropy $k$ is finite, one can check effectively  whether a function is an extractor, and therefore one can effectively construct an extractor with the parameters in Theorem~\ref{t:extopt}. Such an exhaustive search can be done in space $2^{O(n)}$.

Moving to explicit (\ie, polynomial-time computable) extractors, the currently best result for extractors with $O(\log 1/\epsilon$)  entropy loss is  due to  Guruswami, Umans, and Vadhan~\cite{guv:j:extractor}:

\begin{theorem}
\label{t:extpoly}
For all $n$, $k \leq n$, and $\epsilon > 0$, there exists an explicit  $(k,\epsilon)$ extractor with $m = k + d - 2 \log (1/\epsilon) - O(1)$ and $d = \log(n) + O (\log k \cdot \log (k/\epsilon))$.
\end{theorem}
\smallskip

\emph{Note.} In case $k = \Omega(n)$ (and constant $\epsilon$), the extractor in Theorem~\ref{t:extpoly} has $d = O(\log^2 n)$. This is the source of the
$O(\log^2 n)$ terms in our main result Theorem~\ref{t:mainpoly}. It is a major open problem to obtain explicit extractors with $O(\log 1/\epsilon)$ entropy loss that have $d = O(\log n)$.
Such extractors would reduce the overhead in our result to $O(\log n)$.

\section{The upper bounds}
\label{s:ub}

The following two theorems are the main results.
 \begin{theorem}
\label{t:main}
There exists a probabilistic algorithm that on input $x \in \zo^n$ and rational $0 < \delta < 1$, outputs a list with $n$ elements which with probability at least $(1-\delta)$ contains a $O(\log (n/\delta))$-short program for $x$.

Moreover, the algorithm uses $O(\log (n /\delta))$ random bits and can be executed in space $2^{O(n)}$.
\end{theorem}
\begin{theorem}
\label{t:mainpoly}
There exists a probabilistic polynomial-time algorithm that on input $x \in \zo^n$ and rational $0 < \delta < 1$, outputs a list with $n$ elements which with probability at least $(1-\delta)$ contains a $O(\log^2(n/\delta))$-short program for $x$.

Moreover, the algorithm uses $O(\log^2 (n /\delta))$ random bits.
\end{theorem}

The proofs of these two results have a common structure. The key part is  building a bipartite graph with the  ``rich owner" property, roughly meaning  that no matter how we restrict the left side to a subset of  a certain size, then,  in the restricted graph,  most left nodes ``own" most of their neighbors (in the sense that these neighbors are not shared with any other node).

We start by defining precisely the ``rich owner'' property in a bipartite graph~$G$. Let $B$ be a subset of left nodes.  
We say that a right node $y$ is \emph{shared} in $B$ if it has at least two neighbors in~$B$.
For any $\delta>0$, a left node $x$ is $\delta$-\emph{rich} in $B$ 
if $x \in B$, it has at least one right neighbor,
and at most a fraction $\delta$ of its neighbors are shared in $B$
(so it ``owns'' at least a fraction $1-\delta$ of its neighbors). Later, we 
also use a refined version of these concepts.
We say that a right node is  {\em $s$-shared} in $B$ if it has at least $s$ left neighbors in $B$.
A left node $x$ is $(s,\delta)$-{\em rich} in $B$ if $x \in B$, it has at least one neighbor, 
and if at most a fraction $\delta$ of its neighbors are $s$-shared in $B$. 
If the set $B$ is omitted in these definitions, $B = L$ is assumed.

\begin{definition}
A bipartite graph $G = (L,R,E)$ has the \emph{rich owner} property for parameters $(\ell, c,
\delta)$ if for any left  subset $B$ of size at
most $2^\ell $, all but at most $2^{\ell-c}$ of its elements are $\delta$-rich in $B$.
\end{definition}

For us the key parameters are the left degree $D=2^d$ and $m = \log |R|$, because $d$ corresponds to the number of random bits used in the main results and $m - \ell$ essentially gives
the ``quality" of the short program (\ie, the distance between its length and $C(x)$).
The following theorem, whose proof we defer for Section~\ref{s:construction}, shows the existence of this type of graphs with parameters that will allow us to establish the main results.\footnote{Such graphs can be obtained from the extractor-condenser pairs 
from~\cite{raz-rei:c:extcon}, Theorem 5.1. We give here a similar but slightly more efficient construction, with a  proof tailored for our purposes.}

\begin{theorem}\label{cor:extactorsWithRichElements}
  (1)  (Constructible graphs with the rich owner property.)  For all $n,\ell,c,\delta>0$ there exists a family of 
  graphs $G_{n,\ell} = (L= \zo^n,R = \zo^m,E)$ 
  
  which have the rich owner property for  $(\ell, c, \delta)$, with 
\[
\begin{array}{ll}
m & =  \ell + O(c + \log(n/\delta))  \mbox{ and}   \\

d & =  O(c + \log(n/\delta))
\end{array}
\]
The family is uniformly computable in $n,\ell,c,\delta$.
  \smallskip

(2)    (Explicit graphs with the rich owner property.) For all $n,\ell,c,\delta>0$ there exists a family of 
graphs $G_{n,\ell} = (L= \zo^n, R = \zo^m, E)$ which have the rich owner property for  $(\ell, c, \delta)$, with 
\[
\begin{array}{ll}
m &= \ell + O(c + \log^2 (n/\delta)) \mbox{ and} \\
d &=  O(c + \log^2(n/\delta)).
\end{array}
\]
The $y$-th neighbor $E(x,y)$ of $x$ (in the corresponding graph in the family) is computable in time $\poly(n,c, 1/\delta)$.
\end{theorem}

Equipped with graphs that have the rich owner property, we can prove our main results.

\begin{proof}[of Theorem~\ref{t:main}]
  We start by showing a weaker claim:
 
  \begin{claim}\label{c:fixedlength}
    There exists 
    a probabilistic algorithm that on input $x,\ell,c,\delta>0$ always terminates and outputs a
    program of length $\ell + O(c + \log (n/\delta))$, such that 
    for all $\ell,c$ and $n$, for all but at most $2^{\ell-c}$ strings $x$ of length $n$ with $C(x)
    < \ell$, 
    with probability $1-\delta$ the algorithm outputs a program that computes $x$. Moreover, the
    probabilistic algoritm uses $O(\log (c + n/\delta))$ random bits. 
  \end{claim}
 
  If the algorithm was only required to terminate if $C(x) < \ell$, the claim would be easy:
  run all programs of length less than $\ell$ in parallel,
  wait until a program for~$x$ appears, and output
  this program; but this procedure never terminates if $C(x) \ge \ell$.

  To show the claim, fix some $\ell$, $c$, $\delta$, and $n$. 
  Let $B_{n,\ell}$ be the set of all  $n$-bit strings~$x$ with $C(x) < \ell$.
  Note that $B_{n,\ell}$ can be enumerated uniformly in $n$ and $\ell$ and that $|B_{n, \ell}| < 2^\ell $.
  Let $G_{n,\ell}$  
  be the graph satisfying the conditions of Theorem~\ref{cor:extactorsWithRichElements}(1). 
  Thus all but at most $2^{\ell-c}$ nodes are $\delta$-rich in $B_{n,\ell}$. 

  Consider a machine that given an encoding of $\ell,c,  \delta, n$ and a 
  right node $z$ of $G_{n,\ell}$ does the following: 
  it enumerates all $x \in B_{n,\ell}$  and 
  when the first such neighbor $x$ of $z$ in $G_{n,\ell}$  appears, it outputs $x$ and halts.
  All but at most $2^{\ell-c}$ such nodes $x$ are $\delta$-rich in $B_{n,\ell}$, 
  and for such~$x$ a fraction $(1-\delta)$ of neighbors are associated to programs for~$x$ as described above.
  The associated programs can be assumed to have length $\ell + O\left(c+\log (n/\delta)\right)$.
  On input $x$ of length $n$, and $\ell,c, \delta$, the algorithm of the claim, using $O(\log n/\delta)$ random bits, 
  randomly chooses a right neighbor $z$ of $x$ and outputs its associated program. 
  (Note that there is always at least one neighbor.)
  It remains to convert a program on this 
  special machine to a program for our standard reference machine~$U$. 
  This is possible using the function $t$  in the definition of standard machines; 
  it increases the length by $O(1)$ 
  and its computation time by a polynomial function. The claim is proven.~\qed
\medskip

We now proceed to the proof of  Theorem~\ref{t:main}. 
  Let $e$ be a constant such that $C(x) < |x| + e$ for all $x$. 
  For some $x$ and $c$, we could apply the algorithm of the claim for $\ell = e+1, e+2, \dots, |x|+e$ with the same choice of random bits at each iteration (so that the number of random bits remains $O(\log n/\delta)$). 
  In this way we obtain a list of $|x|$ programs such that, with probability $1-\delta$, one of them
  computes~$x$. This has almost the desired effect, the only problem being  that the construction
  may fail on  a $2^{-c}$ fraction of strings $x$ of some length (namely, on the strings that are not rich owners).

  To handle this, we modify the definition of $B_{n,\ell}$ above. The idea is that now $B_{n,\ell}$ should not only 
  contain strings of small complexity, but also the few strings that are not $\delta$-rich in $G_{n,\ell+1}$. 
  More precisely, fix some $n$ and apply Theorem~\ref{cor:extactorsWithRichElements}(1)  with $c = 2$. 
  For $\ell = e+2,\dots, n+e+1$ let $B_{n,\ell}$ be the union of the set of all $n$-bit $x$ 
  with $C(x) <\ell-1$ (first type) and those strings that are 
  not $\delta$-rich in $G_{n,\ell+1}$ (second type). By downward induction we show that $B_{n,\ell}$ can be 
  enumerated from~$n$ and~$\ell$, 
  and that the size of $B_{n,\ell}$ is bounded by~$2^\ell $. 
  Indeed, $B_{n,n+e+1}$ only contains strings of the first type and thus it satisfies the conditions. 
  Now assume the conditions hold for some $\ell \le n+e+1$. Both types of strings can be enumerated.  
  Moreover, the number of strings in $B_{n, \ell-1}$ of the first type is bounded by $2^{\ell-2}$. 
  By the induction hypothesis, $B_{n,\ell}$ has size at most $2^{\ell}$, 
  thus the number of strings that are not rich in $G_{n,\ell}$ is at most $2^{\ell-c} = 2^{\ell-2}$. 
  Hence, the size of $B_{n,\ell - 1}$ is at most $2^{\ell-2} + 2^{\ell-2} = 2^{\ell-1}$.
  This modification only changes the programs associated to right nodes by some fixed instructions,  
  and this does not affect their length by more than a $O(1)$ constant (and the time to generate
  them by more than a polynomial factor). 
\end{proof}

\begin{proof}[of Theorem~\ref{t:mainpoly}]  It is the same proof as above, except that we  use
  Theorem~\ref{cor:extactorsWithRichElements} (2).
  For later reference, we state explicitly  the polynomial-time version of claim~\ref{c:fixedlength}.
  \begin{claim}\label{c:fixedlengthPolytime}
    There exists 
    a probabilistic algorithm that on input $x,\ell,c,\delta>0$ outputs in polynomial
    time a program of length $\ell + O(c + \log^2 (n/\delta))$, such that 
    for all $\ell,c$ and $n$, for all but at most $2^{\ell-c}$ strings $x$ of length $n$ with $C(x)
    < \ell$, 
    with probability $1-\delta$ the algorithm outputs a program that computes $x$. Moreover, the
    probabilistic algoritm uses $O(c + \log^2 (n/\delta))$ random bits.~\qed
  \end{claim}
\end{proof}

\medskip

\section{Construction of graphs with the rich owner property}
\label{s:construction}

All that is left is to prove Theorem~\ref{cor:extactorsWithRichElements}, \ie, to show the
construction of graphs with the rich owner property. 
We use the concept of a
$(s,\delta)$-{\em rich} node in $B$, introduced earlier  (also recall our convention that in case
the set $B$ of nodes  is omitted, it is assumed to be $L$, the set of left nodes).

Our proof has four steps which we describe roughly. 
First we show that most of the left nodes in an extractor graph share their right neighbors with a small number of other
left nodes; i.e., most left nodes are $(s,\delta)$-rich for appropriate~$s$ 
and~$\delta$ (Lemma~\ref{lem:extractorsAre_c_rich}). 
In the second step of the proof we split right nodes (and edges) in a graph such that right nodes share
less left neighbors.  
We do this in a way such that any $(s,\delta)$-rich node in some left subset becomes a $2\delta$-rich
node (Lemma~\ref{lem:c_righ_to_rich}). 
In the third step, we combine the previous results to show that within a small computational cost, 
extractors can be converted to graphs with the rich owner property 
(Proposition~\ref{prop:extractorsWithRichElements}). 
Finally, the theorem is proven using extractors from Theorems~\ref{t:extopt}
and~\ref{t:extpoly}.


\begin{lemma}\label{lem:extractorsAre_c_rich} 
  Let $0 < \varepsilon < 1$. All but at most $2^k$ left nodes of a $(k,\varepsilon)$ extractor with average right degree at
  most $a$ are $(a/\varepsilon,2\varepsilon)$-rich. 
\end{lemma}
\if01
\begin{lemma}\label{lem:extractorsAre_c_rich}
  All but at most $2^k$ left nodes of a $(k,1/c)$ extractor with average right degree at most $a$ are $(ac,2/c)$-rich. 
\end{lemma}
\fi

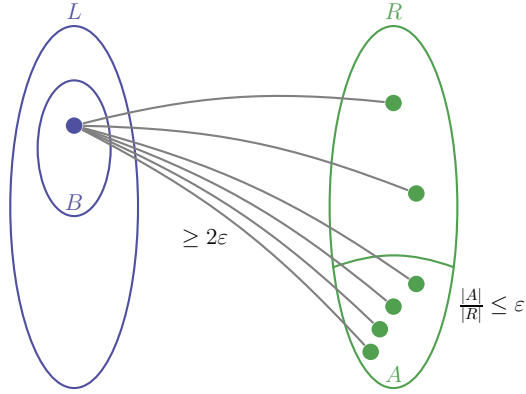
\begin{figure}[h]
  \centering
  \begin{tikzpicture}[scale=0.6, transform shape, thick,
      lNode/.style={fill=myblue, fill, circle},
      rNode/.style={fill=mygreen, fill, circle},
      every fit/.style={ellipse,draw,inner sep=-2pt,text width=2cm},
    ]

    \node[lNode] (l1) at (0,-.5) {};

    \node[draw, ellipse, minimum height=8cm, minimum width=28mm,myblue] (leftset) at (0,-2.3) {};
    \node[anchor=south,myblue] at (leftset.90) {\Large $L$};

    \node[draw, ellipse, minimum height=3cm, minimum width=16mm,myblue] (setB) at (0,-1) {};
    \node[anchor=south,myblue] at (setB.-90) {\Large $B$};

    \begin{scope}[xshift=7cm]
      \node[rNode] at (0,0) (r1)  {};
      \node[rNode] at (.5,-2) (r2)  {};
      \node[rNode] at (-0.5,-5.5) (r6)  {};
      \node[rNode] at (-0.3,-5) (r5)  {};
      \node[rNode] at (-0,-4.5) (r4)  {};
      \node[rNode] at (0.5,-4) (r3)  {};
      \node[draw, ellipse, minimum height=8cm, minimum width=28mm,mygreen] (rightset) at (0,-2.3) {};
      \node[anchor=south,mygreen] at (rightset.90) {\Large $R$};
    \end{scope}

    \path[mygreen] (rightset.-45) edge [bend right=20] (rightset.225);
    \node[mygreen, anchor=south] at (rightset.-90) {\Large $A$};
    \node[anchor=north west] at (rightset.-52) {\Large $\frac{|A|}{|R|} \le \varepsilon$};

    \foreach \i in {1,...,5}
       \path[gray] (l1) edge [bend left=10] (r\i);

    \path[gray] (l1) edge [bend left=10] node[midway,black,anchor=north east] {\Large{$\ge 2\varepsilon$}} (r6);

  \end{tikzpicture}
  \caption{Sets $B$ and $A$ in the proof of Lemma~\ref{lem:extractorsAre_c_rich}.} 
\end{figure}

\begin{proof}
  It suffices to show the lemma for $a$ equal to the average right degree, because for larger $a$ more
  left nodes are $(a/\varepsilon,2\varepsilon)$ rich.
 Let $A$ be the set of right nodes with degree at least $a/\varepsilon$.
 Note that $|A|/|R|$ is at most $\varepsilon$. Let $B$ be the set of left nodes that are not
 $(a/\varepsilon,2\varepsilon)$ rich, 
 \ie, have more than a fraction $2\varepsilon$ of neighbors in~$A$. 
 Consider a flat distribution over all edges leaving from
 $B$. Inequality \eqref{eq:extractorDef} in the definition of a $(k,\varepsilon)$ extractor is 
 violated:
 \[
\bigg| \prob[E(X, U_d) \in A] - \frac{|A|}{M} \bigg|
  > \left| 2\varepsilon - \varepsilon \right| = \varepsilon\,.
 \]
 This implies that $B$ has less than $2^k$ elements, 
 and this implies the lemma.
\end{proof}

In the second step, we split edges and right nodes of a graph. This splitting satisfies the
following property for all subsets $B$ of left nodes:  if a right node has few neighbors in $B$, 
then most of the corresponding splitted nodes have no or a unique neighbor in~$B$.
We do this using a technique from~\cite{bfl:j:boundedkolmogorov} (and also employed
in~\cite{bmvz:c:shortlist}). We  hash the right nodes using  congruences modulo a small set of
prime numbers.  For this technique we need the following lemma. 

\begin{lemma}\label{l:mod}
Let $x_1, x_2 \ldots, x_s$ be distinct $n$-bit strings, which we view in some canonical way as integers $< 2^{n+1}$.
Let $p_i$ be the $i$-th prime number and let $L = \{p_1, \ldots, p_t\}$, where $t = (1/\delta) \cdot s \cdot n$.

For every $i \leq s$, for less than a fraction $\delta$ of $p$ in $L$, 
the value of $x_i\bmod p$ appears more than once in the sequence $(x_1\bmod p,x_2\bmod p,\ldots,x_s\bmod p)$.
\end{lemma}

\begin{proof}
By the Chinese Remainder Theorem and taking into account that the $x_i$ values are bounded by $2^{n+1}$, we notice that, for every $x_j \not= x_i$, ``$x_i = x_j \hspace{-0.2cm}\mod p$"   holds for at most $n$ prime numbers $p$. Therefore,  ``$\exists x_j \not= x_i  ( x_i = x_j\hspace{-0.2cm} \mod p)$"  holds for at most $(s-1)n$ prime numbers $p$. Since $L$ contains $(1/\delta) \cdot s \cdot n$ prime numbers, it follows that
\[
  \prob_{p \in L}[\mbox{$x_i \hspace{-0.2cm}\mod p$ is not unique}] \leq \frac{(s-1)n}{(1/\delta)
  \cdot s \cdot n} < \delta. 
\]
\end{proof}

\begin{lemma}\label{lem:c_righ_to_rich}
  For any $s$, $\delta>0$ and (left-regular) graph $G = (L = \zo^n,R,E)$, let $t = sn/\delta$.
  There exists a (left-regular) graph $H = (L, R \times S, E')$  such that 
  \begin{enumerate}

 \item The left degree of $H$  is exactly $t$ times the left degree of $G$, 
  
 \item The set $S$  has size at most $O\left(t^3\right)$,
  
 \item   If in $G$ the $i$-th right node of a left node $n$  can be constructed in $time (n)$,  the same
 operation in $H$ takes $time(n) + \poly(t)$, 
  and

\item  If a left node $x$ is $(s,\delta)$-rich in some $B$ for $G$, 
       then it is $2\delta$-rich in $B$ for $H$. 
  \end{enumerate}
\end{lemma}

\begin{figure}[h]
  \hspace{-5mm}
  \begin{tikzpicture}[scale=0.5, transform shape,
      lNode/.style={fill=myblue, circle},
      rNode/.style={fill=mygreen, circle}
      ]

      \clip (-2,3) rectangle (9,-8);

    \node[draw, ellipse, minimum height=8cm, minimum width=28mm,myblue] (leftset) at (0,-2.5) {};
    \node[anchor=south,myblue] at (leftset.90) {\Large $L$};

    \node[lNode] (l1) at (0,-.5) {};
    \node[left of=l1, xshift=5mm]  {\Large $x_1$};
    \node[lNode] (l2) at (0,-2.5) {};
    \node[left of=l2, xshift=5mm]  {\Large $x_2$};
    \node[lNode] (l3) at (0,-4.5) {};
    \node[left of=l3, xshift=5mm]  {\Large $x_j$};

    \begin{scope}[xshift=7cm]
      \node[draw, ellipse, minimum height=16cm, minimum width=28mm,mygreen] (leftset) at (0,-6.5) {};
      \node[anchor=south,mygreen] at (leftset.90) {\Large $R$};

      \node[rNode] (r1) at (0,-0.5) {};
      \node[rNode] (r2) at (0,-5.5) {};
      \node[rNode] (r3) at (0,-10.5) {};
      \node[rNode] (r4) at (0,-15.5) {};
    \end{scope}

    \foreach \i in {1,...,3}
       \path[gray] (l\i) edge [bend left=10] (r1);

    \foreach \i in {1,...,4}
       \path[gray] (l3) edge [bend left=10] (r\i);

    \node[right of=r1, xshift=-5mm]  {\Large $z$};
    \end{tikzpicture}
    \quad \quad \quad
    \begin{tikzpicture}[scale=0.5, transform shape,
      lNode/.style={fill=myblue, circle},
      rNode/.style={fill=mygreen, circle}
      gNode/.style={fill=gray, circle}
      ]

      \clip (-2,3) rectangle (12,-8);

    \node[draw, ellipse, minimum height=8cm, minimum width=28mm,myblue] (leftset) at (0,-2.5) {};
    \node[anchor=south,myblue] at (leftset.90) {\Large $L$};

    \node[lNode] (l1) at (0,-.5) {};
    \node[left of=l1, xshift=5mm]  {\Large $x_1$};
    \node[lNode] (l2) at (0,-2.5) {};
    \node[left of=l2, xshift=5mm]  {\Large $x_2$};
    \node[lNode] (l3) at (0,-4.5) {};
    \node[left of=l3, xshift=5mm]  {\Large $x_j$};
     
    \begin{scope}[xshift=7cm]
      \node[draw, rectangle, rounded corners, minimum height=16cm, minimum width=5cm,mygreen] (rightset) at (0,-6.5) {};
      \node[anchor=south,mygreen] at (rightset.90) {\Large $R$};

      \foreach \i/\j in {0/1,-0.5/2}
	 \fill[gray] (-1.5,\i) node (r1\j) {} circle (0.1);

      \foreach \i/\j in {0/1,-0.5/2,-1/3}
	 \fill[gray] (-.5,\i) node (r2\j) {} circle (0.1);

      \fill[red] (r22) {} circle (0.1);

      \foreach \i/\j in {0/1,-0.5/2,-1/3,-1.5/4, -2/5}
	 \fill[gray] (.5,\i) node (r3\j) {} circle (0.1);

      \foreach \i/\j in {0/1,-0.5/2,-1/3,-1.5/4, -2/5, -2.5/6, -3/7}
	 \fill[gray] (1.5,\i) node (r4\j) {} circle (0.1);

      \foreach \i in {1,...,4}
	\node[above of=r\i1,anchor=north]  {\Large $p_{\i}$};

      \node[rNode] (r2) at (0,-6) {};
      \node[rNode] (r3) at (0,-12) {};
      \node[rNode] (r4) at (0,-18) {};

      \draw[decorate,decoration={brace, amplitude=6pt},gray] ($(r41) + (1.5,.4)$) --
      node[anchor=west,xshift=5mm] {\Large $z$} ($(r47) + (1.5,-.4)$);
    \end{scope}

    \coordinate[xshift=-4mm] (b1) at (rightset.110);
    \coordinate[xshift=-4mm] (b2) at (rightset.114);
    \coordinate[xshift=-4mm] (b3) at (rightset.118);

    \path[gray,ultra thick] (l1) edge [bend left=10] (b1);
    \path[gray,ultra thick] (l2) edge [bend left=10] (b2);
    \path[gray,ultra thick] (l3) edge [bend left=10] (b3);

    \path[gray] ($(b1)+(0,-1pt)$) edge [bend left=10] (r11);
    \path[gray] (b1) edge [bend left=10] (r21);
    \path[gray] (b1) edge [bend left=10] (r32);
    \path[gray] ($(b1)+(0,1pt)$) edge [bend left=10] (r41);

    \path[gray] ($(b2)+(-1pt,1pt)$) edge [bend left=5] (r11);
    \path[gray] (b2) edge [bend left=25] (r22);
    \path[gray] (b2) edge [bend left=10] (r33);
    \path[gray] ($(b2)+(0,-1pt)$) edge [bend left=10] (r45);

    \path[gray] ($(b3)+(0,1pt)$) edge [bend left=10] (r12);
    \path[red] (b3) edge [bend left=10] (r22);
    \path[gray] (b3) edge [bend left=10] (r35);
    \path[gray] ($(b3)+(0,-1pt)$) edge [bend left=10] (r46);

    \begin{scope}[xshift=7cm,yshift=-5.3cm]
      \foreach \i/\j in {0/1,-0.5/2}
	 \fill[gray] (-1.5,\i) node (r1\j) {} circle (0.1);

      \foreach \i/\j in {0/1,-0.5/2,-1/3}
	 \fill[gray] (-.5,\i) node (r2\j) {} circle (0.1);

      \foreach \i/\j in {0/1,-0.5/2,-1/3,-1.5/4, -2/5}
	 \fill[gray] (.5,\i) node (r3\j) {} circle (0.1);

      \foreach \i/\j in {0/1,-0.5/2,-1/3,-1.5/4, -2/5, -2.5/6, -3/7}
	 \fill[gray] (1.5,\i) node (r4\j) {} circle (0.1);
    \end{scope}

    \coordinate[xshift=-4mm] (b1) at (rightset.145);
    \coordinate[xshift=-4mm] (b2) at (rightset.-140);
    \coordinate[xshift=-4mm] (b3) at (rightset.-110);

    \foreach \i in {1,...,3}
       \path[gray,ultra thick] (l3) edge [bend left=10] (b\i);
    \path[red,ultra thick] (l3) edge [bend left=10] (b2);

    \path[gray] ($(b1)+(0,-1pt)$) edge [bend left=10] (r12);
    \path[gray] (b1) edge [bend left=10] (r21);
    \path[gray] (b1) edge [bend left=10] (r34);
    \path[gray] ($(b1)+(0,1pt)$) edge [bend left=10] (r41);
  \end{tikzpicture}
  
  \caption{Splitting of right nodes in the proof of Lemma~\ref{lem:c_righ_to_rich}.}
  \label{fig:splittingNodes}
\end{figure}
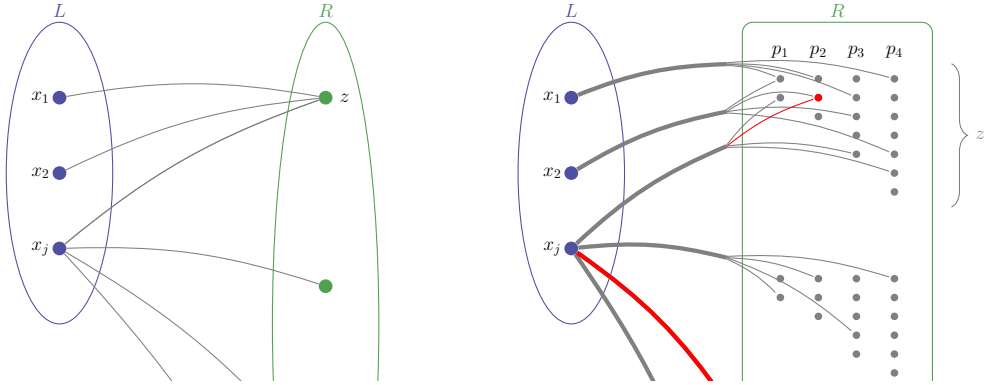

\begin{proof}
  Let $G = (L,R,E)$ and let $p_1, p_2, \ldots, p_t$ be the first $t$ prime numbers.
  The right set of $H$ is given by
\[  
   R \times\{p_1, \ldots, p_t\} \times \{0,1, \ldots, p_t-1\}\,.
\]
 The edges of $H$ are obtained by adding for each edge $(x,z)$ in $G$ the edges 
\[
\begin{array}{l}
(x, (z, p_1, x\bmod p_1)),  (x, (z, p_2, x\bmod p_2)),\ldots, 
(x, (z, p_t,  x \bmod p_t))
\end{array}
\]
 in $H$ (one can think that each edge  $(x, z)$ in  $G$ is split into $t$  edges in $H$, see
 Figure~\ref{fig:splittingNodes}).

 This operation increases the left degree by a factor of $t$, which implies (1). The size of $S = \{p_1, \ldots, p_t\} \times
 \{0,1, \ldots, p_t-1\}$ is bounded by ${p_t}^2$. 
 Because $p_t \le t \ln t + t \ln\ln t$ for $t \ge 6$, this size is bounded by $O(t^3)$, which implies (2).
 Enumerating the first $t$ prime numbers is possible in time $\poly(t)$, which implies (3).

 It remains to show why each $(s,\delta)$-rich left node in some $B$ in graph $G$, is $2\delta$-rich
 in $B$ in graph $H$. The proof for general $B$ follows the proof for $B=L$ which is presented here.
 Suppose a right node $z$ in $G$ is not $s$-shared, \ie,  it has  $s'$ neighbors 
 $x_1,\dots, x_{s'}$ with $s' \le s$. 

 For some $j \le s'$,  
 how many nodes $(z, p_i, x_j \bmod p_i)$ with $i = 1, \dots, t$ are shared? 
 If for some $p_i$ the  value $x_j \bmod p_i$ appears more than once in $x_1 \bmod p_i, \ldots, x_{s'} \bmod p_i$, 
 then $(z, p_i, x_j \bmod p_i)$ is shared. By Lemma~\ref{l:mod}, 
 at most a fraction $\delta$ of nodes 
 $(z, p_i, x_j \bmod p_i)$ in $H$ are shared (red element in figure~\ref{fig:splittingNodes}).

\if01
 For a fixed $i \le t$, how many nodes $(z, p_i, x_j \bmod p_i)$ are shared? 
 If for some $i$ all values $x_j \mod p_i$ with $j \le s'$ are different, 
 then $(z, p_i, x_j \bmod p_i)$ is not shared. By Lemma~\ref{l:mod} this happens for at least a fraction
 $1-\delta$ of $p_i$.  Thus, for fixed $j$ at most a fraction $\delta$ of nodes $(z, p_i, x_j \mod
 p_i)$ in $H$ are shared.
\fi

 Let $x_j$ be a left node in $G$ that is $(s,\delta)$-rich. 
 
 Then at most a fraction $\delta$ of 
 its right nodes in $G$ are $s$-shared (denoted by thick red edges in figure~\ref{fig:splittingNodes}). 
 By the previous paragraph, of the remaining 
 $1-\delta$ fraction of nodes, at most a fraction $\delta$ are shared in $H$. 
 Thus the total fraction of shared nodes is at most $\delta + \delta(1-\delta) \le 2\delta$; 
 \ie, $x_j$ is $2\delta$-rich for $H$.
\end{proof}

Now we combine previous results to show that we can convert extractors to graphs in which 
most left nodes in a sufficiently small set $B$ are $\varepsilon$-rich for some~$\varepsilon$.

\begin{proposition}\label{prop:extractorsWithRichElements}
 Let $E: \zo^n \times \zo^d \mapping \zo^m$ be a $(k,\varepsilon)$
  extractor, and let $G$ be the graph associated to $E$.
   For any natural number $a$,  there exist a set $S$ and a non-empty (left-regular) graph $H = (L,R\times S,E')$ such that 
  \begin{enumerate}
    \item \label{item:increases} 
      the left degree is at most $O(2^dan/\varepsilon^2)$
      and $|S| \le O\left( (an/\varepsilon^2)^3 \right)$,
    \item \label{item:extentionPolyTime}
     from $n, k, \varepsilon, (x,y)$, one can compute a list of all $s$ such that 
     $(x,(y,s)) \in E'$ in time $\poly(n,a,1/\varepsilon)$,
   \item \label{item:richNodes}
     if $B \subseteq \zo^n$ is such that the average number of edges from $B$ arriving in a right node from $G$ is at most~$a$, 
     then all but at most $2^k$ left elements in $B$ are $4\varepsilon$-rich in $B$ relative to~$H$. 
  \end{enumerate}
\end{proposition}

\begin{proof}
  Note that after deleting some left nodes
  from a $(k,\varepsilon)$ extractor graph,  what is left is still a $(k,\varepsilon)$ extractor graph.
  More precisely, for any graph $G=(L,R,E)$, a subset $L' \subseteq L$ defines 
  a subgraph $G'=(L',R,E' \subseteq E)$ where $E'$ are all edges in $E$ leaving from $L'$. 
  If $G$ is an extractor graph, then also $G'$ is an extractor graph: indeed, we must verify
  \eqref{eq:extractorDef} for all $B \subseteq L$ of size at least $2^k$ in $G'$ and $A \subseteq R$, 
  but every element in $B$ has the same edges in $G'$ and hence defines the same probabilities.

  Apply Lemma~\ref{lem:c_righ_to_rich} with $s = a/\varepsilon$ and $\delta = \varepsilon$;  
  thus $t = sn/\delta = an/\varepsilon^2$ and this guarantees 
  the existence of $H$ such that conditions~\ref{item:increases} and~\ref{item:extentionPolyTime} are satisfied.
  Clearly the graph is left regular, and because extractor-graphs are non-empty, this also holds for $H$.
  To show item~\ref{item:richNodes}, we use  Lemma~\ref{lem:extractorsAre_c_rich}. 
  As discussed above, the extractor property remains true if the left set is restricted to a subset $B$, 
  thus we can apply  Lemma~\ref{lem:extractorsAre_c_rich}. 
  Hence, all but at most $2^k$
  nodes in $B$ are $(a/\varepsilon,2\varepsilon)$-rich in $B$, and by the last part of  Lemma~\ref{lem:c_righ_to_rich}, 
  they are $4\varepsilon$-rich in $B$ relative to~$H$.
\end{proof}

Finally, we are ready to finish the proof.
\smallskip

\begin{proof}[of Theorem~\ref{cor:extactorsWithRichElements}]
 We prove (1), the proof for (2) is similar (we use
 the extractor from Theorem~\ref{t:extpoly} instead of  Theorem~\ref{t:extopt}).

 Let $G = (L = \zo^n, R = \zo^m, E)$ be an extractor satisfying the conditions of  Theorem~\ref{t:extopt} and suppose $B$ is a
 subset of $L$ of size at most $2^{k+c}$. Let us compute an upper bound for the average right degree of the
 subgraph containing all edges leaving from $B$: 
 \[
 \frac{|B|D}{2^m} \leq 2^c \frac{2^kD}{2^m}\,
 \]
 \ie, $2^c \times 2^{\mbox{(the entropy loss)}}$.  Since the entropy loss in Theorem~\ref{t:extopt} is $2 \log (1/\varepsilon) + O(1)$, the average right degree is
 $O(2^c/\varepsilon^2)$.

 Next we choose $\delta$ and $\ell$:
 \[
 \delta = 4\varepsilon \text{\quad\quad and\quad\quad} \ell = k + c \,.
 \]
 We apply
  Proposition~\ref{prop:extractorsWithRichElements}  and conclude 
 that all but at most $2^{\ell-c}$ left nodes of $H$ 
 are $\delta$-rich in $B$. 
 Any left node in a non-empty left-regular graph has at least one neighbor, 
 thus the graph has the rich owner properties for parameters $(l,c,\delta)$.
 It remains to check the claimed values for $d$ and $m$.

 In the extractor, the left degree $D \le O(n/\varepsilon^2) = \Theta(n/\delta^2)$.
 The new graph has left degree $O(Dan/\varepsilon^2)$ and 
 \[
 D\cdot a \cdot n/\varepsilon^2 = O\left(\frac{n}{\delta^2} \frac{2^c}{\delta^2}\frac{n}{\delta^2}
 \right)\,.
 \]
 Hence, the logarithm of the left degree is 
  $d = O\left(\log Dan/\varepsilon^2\right) = O\left(c + \log (n/\delta)\right)$.
 
 The logarithm of the size of the set of right vertices is
 \[
  m = \left(k + \log D - 2\log (1/\varepsilon) + O(1) \right) + \log |S|\,.
 \]
 Note that $|S| \le O\left( (Dan/\varepsilon^2)^3 \right)$, so $\log |S|$ is
 $O\left(c + \log (n/\delta)\right)$. Similar for $\log D$ and 
 $\log (1/\varepsilon)$, and thus
 $m = (\ell-c) + O\left(c + \log (n/\delta)\right)$.
\end{proof}

\section{Lower bounds}
\label{s:lb}

There are three parameters of interest in the main result Theorem~\ref{t:main}, which on input $x$ constructs with high probability a list containing a short program for $x$:
\smallskip

 $T$ = the size of the list, 

$r$ = the number of random bits used in the construction, and

$c$ = the closeness to $C(x)$ of the short program guaranteed to exist in most lists.
\smallskip

 In Theorem~\ref{t:main}, we show $T=n$, $r = O(\log n)$, and $c = O(\log n)$, where $n$ is the length of the string $x$ (for simplicity, we have assumed here constant error probability $\delta$). We show here that Theorem~\ref{t:main} is tight, in the sense
that, essentially,  none of these parameters can be reduced while keeping the other two at the same level.

 To prove the lower bounds, we need to specify the model carefully. A \emph{probabilistic algorithm that list-approximates short programs} is given by a Turing machine $F$ that takes an input $x$, and a sequence $\rho$ (of random bits)  and satisfies the following properties:
\smallskip

(a) $|\rho| = r$ depends only on the length of $x$,

(b) for all $x$ and $\rho$ of length $r$, $F(x,\rho)$ halts and outputs a finite set of strings $L_{x,\rho}$ (which we typically call a \emph{list}),

(c)  the size $|L_{x,\rho}| = T$, depends only on the length of $x$.  

(d) for all $x$, at least $1/2$ of the sets 
 $\{L_{x,\rho}\}_{|\rho| = r}$ contain a $c$-short program for $x$. (Since we seek lower bounds, assuming $\delta = 1/2$, implies at least as strong lower bounds for smaller~$\delta$).
\smallskip

 After these preparations, we can state the lower bounds.

 \begin{theorem}
\label{t:lb}
 For any probabilistic algorithm that list-approximates short programs with parameters $T$, $r$ and $c$,
 \smallskip
 
 (1) $r \geq 2 \log n - \log T - 2 \log c - O(1)$. In particular, if $T =n$, and $c = O(\log n)$, then $r \geq \log n - O(\log \log n)$.
 \smallskip
 
 (2) $c \geq \log (n^2/T) - 2 \log \log (n^2/T)$. In particular, if $T = O(n)$, then $c \geq \log n - 2\log \log n - O(1)$.
 \smallskip
 
 (3) $T = \Omega(n/(c+1))$.

 \end{theorem}

 \begin{proof}
  
  (1) The union $L = \bigcup_{|\rho| = r} L_{x,\rho}$ is a computable set with $R \cdot T$ elements, where $R=2^r$, and contains a $c$-short program for $x$.
  It is known from~\cite{bmvz:c:shortlist} that any such set must have size $\Omega(n^2/(c+1)^2)$. Thus for some constant $c_1$,
  $R \cdot T \geq c_1 \cdot(n^2/(c+1)^2)$, and thus $r \geq 2 \log n - \log T - 2 \log c - O(1)$.
  \smallskip
  
  (2)  Let ${\cal P}$ be the set of $c$-short programs for $x$. By a result of Chaitin~\cite{cha:j:shortprog} (see also Lemma 3.4.2 in~\cite{dow-hir:b:algrandom}), we know that $\ell  =|{\cal P}|$ satisfies $\ell = O(2^c)$. Since at least half of the $R = 2^r$  lists $L_{x,\rho}$, with $\rho$ ranging in $\zo^r$,  contain an element of ${\cal P}$, there is some element of ${\cal P}$ that belongs to at least $1/(2\ell)$ fraction of lists. Clearly the length of this element is bounded by $n+d$, where $d$ is a constant that depends on the universal machine.
  In steps $m = 1, 2, \ldots, n+d$, we select all the strings of length $m$ that appear in at least $1/(2 \ell)$ of the lists. As argued above, there exists a $c$-short program for $x$ among the selected elements. Let us estimate how many strings have
  been selected. Let $s_m$ be the number of elements selected at the $m$-th step of the selection procedure.
  The elements selected at step $m$ occur at least $s_m \cdot \frac{R}{2 \ell}$ times in the union $L = \bigcup_{|\rho| = r} L_{x,\rho}$. Since $|L| = R \cdot T$, we obtain
  \[
  R \cdot T \geq s_1 \cdot \frac{R}{2\ell} + s_2 \cdot \frac{R}{2\ell} + \ldots + s_{n+d} \cdot\frac{R}{2\ell}.
   \]
  Thus, $s_1 + s_2 + \ldots + s_{n+d} \leq T \cdot 2 \ell$. By the same result from~\cite{bmvz:c:shortlist}, the total number of selected elements is at least $c_1 \cdot n^2/(c+1)^2$, for some constant $c_1$, because some $c$-short program is selected. Thus, 
  \[
  T \cdot 2 \ell \geq s_1 + s_2 + \cdots + s_{n+d} \geq c_1 \frac{n^2}{(c+1)^2}.
  \]
 It follows that $2^c \geq c'_1 \cdot n^2 \cdot (1/T) \cdot (1/(c+1)^2)$, for some constant $c'_1$. The conclusion follows after some simple calculations.
 \smallskip
 
 (3)  Let $L'_{x,\rho}$ be the set of lengths of strings in $L_{x,\rho}$. We say that an integer $m$ has a \emph{pseudo-presence} in  
 $L'_{x,\rho}$ if at least one of the values $m, m+1, \ldots, m+c$ is in $L'_{x,\rho}$. Clearly, at most $(c+1)T$ integers have a pseudo-presence in $L'_{x,\rho}$. We say that $m$ is \emph{significant} if it has a pseudo-presence in at least half of the sets $L'_{x,\rho}$, $|\rho| = r$.
 Note that the set of significant integers contains $C(x)$, because the union of all lists contains a $c$-short program for $x$. By a result of~\cite{bei-bur:j:enumerations}, any computable set containing $C(x)$ must have
 size at least $n/a$ for some constant $a$. Thus, there are at least $n/a$ significant integers.
 Each significant integers has at least $R/2$ pseudo-presences in the union of all the lists. We
 obtain that $(n/a) \cdot (R/2) \leq R \cdot (c+1)T$, which implies $T \geq (1/(c+1)) \cdot (n/(2a))$.
 \end{proof}

\emph{Note.} The lower bound in (3) holds even for randomized algorithm that may not halt on some probabilistic branches. This can be proved in the same way, taking advantage of the fact that the lower bound in~\cite{bei-bur:j:enumerations} holds true even for for algorithms that \emph{enumerate} a list of possible values for $C(x)$. 

\section{Randomized vs. deterministic computation}
\label{s:randdeterm}

We present 
a non-trivial task that shows the superior power of randomized algorithms versus deterministic algorithms. This task
\smallskip

(1) cannot be solved  by any deterministic algorithm that runs
in  computably bounded time, but
\smallskip

 (2)  can be solved by a randomized algorithm in
polynomial time,  and, furthermore, the  number of random bits is only polylogarithmic.
\smallskip

\noindent
Moreover, the task is natural in the sense that it is not artificially constructed to satisfy the
conditions.  
The task is a promise problem, meaning that the input is guaranteed to satisfy a certain property.
\medskip

TASK T: On input $(x, C(x))$, 
generate a $c(|x|)$-short program for $x$. 
\medskip

 If the input is of the form $(x,\ell)$  with $\ell \not= C(x)$ (\ie, the promise is not satisfied), then no output needs to
be generated or the output can be an arbitrary string. Note that in case the promise is satisfied,
then  by exhaustive search  we can  always find a $c(|x|)$-short program even with $c(|x|)=0$.
However,  we will see  in Lemma~\ref{l:noComputablyBounded}, that, for some  $c(n) = O(n)$,  no deterministic algorithm can succeed in computably bounded time.

First we show that for some $c(n) = O(\log^2 n)$ task T can be solved in randomized polynomial 
time using only a polylogarithmic number of random bits.

\if01
\begin{theorem}
There exists a polynomial-time randomized algorithm that on input $(x,\ell, 1/\delta)$, where $x$ is an $n$-bit string, and $\ell$ and $1/\delta$ are positive integers, returns a string $z$. If $C(x)=\ell$, then with probability
$(1-\delta)$, $z$ is a $O(\log^2 (n/\delta))$-short program for $x$. Furthermore, the algorithm uses $O(\log^2 (n/\delta))$ random bits.
\end{theorem}
\fi
\begin{theorem}
\label{t:probalg}
There exists a polynomial-time randomized algorithm that on input $(x,C(x),1/\delta)$ returns a string $z$ that, with probability
$(1-\delta)$, is  a $O(\log^2 (|x|/\delta))$-short program for $x$. Furthermore, the algorithm uses
$O(\log^2 (|x|/\delta))$ random bits.
\end{theorem}
\begin{proof}
The result is a corollary of Claim~\ref{c:fixedlengthPolytime}. Let $n$ be the length of $x$ and we
apply the claim with $\ell = C(x)+1$ and $c=3\log n$. The algorithm in the claim provides  

a $O(\log^2 (n/\delta))$-short program for $x$ unless $x$ belongs to some ``bad'' set of at most $2^{\ell - c}$
strings. We show that for large~$n$ and~$\ell$, strings in this set satisfy $\ell > C(x) + 1$,
hence, $x$ can not belong to this set. 

Indeed, apply the algorithm of the claim on all strings of length~$n$ and all random seeds of sufficient
length. Subsequently, search for $2^n - 2^{\ell-c}$ strings~$x$ for which at least a fraction $1-\delta$
of the generated programs print~$x$. The bad strings are contained in the remaining $2^{\ell - c}$
strings. Hence, such strings $x$ satisfy
\[
C(x) \leq \ell - c + 2\log n + 2\log c + 2\log (1/\delta) < \ell-1. 
\]
\end{proof}
Next, we show that task $T$ cannot be solved by an algorithm in computably bounded time.

\begin{lemma}\label{l:noComputablyBounded}
  For each standard machine $U$ there exists an $e$ such that 
 there is no computable function $t(n)$ and no 
 algorithm that on input 
 $(x, C_U(x))$ outputs a $|x|/e$-short program for $x$ in time $t(|x|)$.
\end{lemma}

\begin{proof}
 Suppose the lemma was false, and for some value $2e$ there exists a computable function $t$ 
 that bounds the computation time of such an algorithm computing an $|x|/(2e)$-short program. 
 Then, there is a total computable function $f$ such
 that for all $x$, the value $f(x,C(x))$ is a $|x|/(2e)$-short program.
 We can assume $|f(x,k)| \le k + |x|/(2e)$, because if the length were larger, we could 
 replace it by any shorter string without affecting the assumption on $f$.

 Fix some large length $n$ 
 and consider for all $x$ the list 
 \[
 \left[ f(x, n/(2e) - (e/2 - 1)), \; f(x, n/(2e) - (e/2 - 2)), \; \dots, \;f(x, n/(2e) + (e/2 - 1)) \right].
 \]
 This list contains at most $e-1$ strings and the sum of the lengths is at most $n(e-1)/e = n - n/e$.
 Therefore, there are at most $2^{n - n/e}$ distinct lists. In the sequence of  $2^n$ lists that correspond to
 strings $x$ of length $n$,  we take $L$ to be the list that appears most often (if there is a tie, we break it in some canonical way).  
 Let $S$ be the set of strings $x$ that generate $L$. It follows from above that $S$ has size at least $2^{n/e}$.

 Note that $S$  can be computed from $n$ and $e$. 
We show that for large $e$ the set $S$ contains $e$ strings, $x_1, x_2, \ldots, x_e$ that have complexity between $n/(2e) -(e/2-1)$ and
$n/(2e) + (e/2-1)$. This leads to a contradiction because, by the assumption on $f$, these $e$ distinct strings must each have a program in $L$, but
$L$ has only at most $e-1$ members.

Let $w$ be a string of length $n/(2e)$ with $C(w\mid n, e) \ge
 n/(2e)$ (such a string exists by a counting argument). 
We  interpret $w$ as a natural number (bounded by $2^{n/(2e)}$), and for $i = 1, \dots, e$, 
 let $x_i$ be the lexicographically $(w+i)$-th element of $S$.

 We now prove that for all $i$, $C(x_i) = n/(2e)  \pm  O(\log e)$.
\if01
 Note that 
 \[
   C(x|k,e) = k + l \quad \text{  implies  }\quad   C(x|e) = k + O(l), 
 \]
 (indeed, $C(x|C(x)) = C(x) + O(1)$ thus $C(x|k) - k = C(x) - k + O(|\log C(x) - k|)$, 
 i.e. $|C(x|k) - k| \le O(|C(x) - k|)$ and the same holds with $e$ in the condition).
\fi
 Because $i \le e$, we have that $C(w\mid n,e) = C(x_i \mid n,e) + O(\log e)$.
 The left hand side equals $n/(2e) \pm O(1)$, hence $C(x_i\mid n/(2e), e) = n/(2e) \pm O(\log e)$. 
We can eliminate the $n/(2e)$ in the conditioning by appending to the program that prints $x_i$ a self-delimiting string
of length $O(\log e)$ that represents the difference between the length of the program and $n/(2e)$ (so that the 
quantity $n/(2e)$ present in the condition can now be read from the modified program). Thus, $C(x_i \mid e) = n/(2e) \pm O(\log e)$, and therefore, as claimed,
$C(x_i) = n/(2e) \pm O(\log e)$.

 Choose $e$ large enough such that the  value of the $O(\log e)$ term 
 in the above equation is at most $e/2 - 1$. (Note that the  constant hidden in $O(\log e)$ 
 is machine dependent, and, consequently, so is $e$.) This implies, that indeed  $\left| C(x_i) - n/(2e) \right| \le (e/2 -1)$, which,
 as we have seen, implies a contradiction.
\end{proof}
\smallskip

\emph{Note.} Our methods show that randomized algorithms running in polynomial time can be more powerful than deterministic machines in any level of the arithmetic hierarchy at solving 
non-trivial tasks. For any $m \geq 0$, let $\emptyset^{(m)}$ denote as usual the $m$-th jump of the Halting Problem and let $C^{(m)}(x)$ denote the length of the shortest program that
 prints the string $x$ relative to a fixed standard oracle universal machine $U$ that uses $\emptyset^{(m)}$ as an oracle. Thus $C^{(0)}(x) = C(x)$, the usual Kolmogorov complexity. We define the
notion of \emph{$c$-short $\emptyset^{(m)}$-program} for a string $x$ similarly to $c$-short programs except that we allow $U$ to use the oracle $\emptyset^{(m)}$. Let us consider
the task T$^{(m)}$ (defined analogously to the task $T$):
\smallskip

TASK T$^{(m)}$: On input $(x, C^{(m)}(x) )$, generate a $c(|x|)$-short  $\emptyset^{(m)}$-program for $x$. 
\smallskip

Theorem~\ref{t:probalg} and Lemma~\ref{l:noComputablyBounded} can be relativized in the straightforward way to show that for any $m \geq 0$,
\smallskip

(a) Task T$^{(m)}$ with parameter $c(n) = O(\log^2 n)$ can be solved by a randomized algorithm (without any oracle) that runs in polynomial time and uses
$O(\log^2 n)$ random bits (where $n = |x|$).
\smallskip

(b) Task T$^{(m)}$, for  some parameter $c(n) = O(n)$, cannot be solved by any deterministic algorithm that uses $\emptyset^{(m)}$ as an oracle and runs in time
bounded by any $\emptyset^{(m)}$-computable function.
%

\bibliography{c:/book-text/theory}
\bibliographystyle{ACM-Reference-Format-Journals}

\appendix
\section{Another non-computable task with trivial randomized solution} 

The task is to select for any $x$ a value $0 < k \le |x|$ such that $|C(x) - k| \ge 2\log |x|$.
It is probabilistically solvable with probability error at most $\epsilon$ by using only a constant amount of random bits: 
Choose at random an integer $r$ in the range $0 < r \le 1/\epsilon$, and return $\varepsilon r|x|$.
By the following lemma this task is not deterministically solvable.

  \begin{lemma}\label{lem:nonCompTask}
    There exist no computable function $f$ such that $0 < f(x) \le |x|$ and such that $|f(x) - C(x)|
    \ge 2\log |x|$ for all $x$.
  \end{lemma}
  \begin{proof}
   Suppose such a computable function exists. Choose some large $n$ 
   and consider the set of values $f(x)$ for all $x$ of length $n$. 
   There is a value $i$ that appears at least $2^n/n$ times and let $S$ be the 
   set of $x$ of length $n$ such that $f(x)$ equals $i$. Thus, $|S| \ge 2^n/n$. 
   To obtain a contradiction, we consider two cases. Assume $i \ge n-\log n$. 
   Because $|S| \ge 2^n/n$, the set contains a string $x$ with $C(x) \ge n - \log n$. 
   Because $C(x) \le n+O(1)$, this implies $|i - C(x)| \le \log n + O(1)$, and because 
   $f(x) = i$ this violates the assumption on $f$ for large $n$.
   Now suppose $i < n - \log n$. Thus $2^i \le 2^n/n \le |S|$. 
   Consider the set of the lexicographic first $2^i$ strings in $S$
   This set can be computed from $n$ and $i$ (and $i$ can be computed from $n$ using $f$). 
   Hence, there is an $x$ such that $C(x|n) = i + O(1)$. 
   This implies $|C(x) - i| \le \log n + O(\log\log n)$ and by construction $f(x) = i$ 
   and this violates the assumption on $f$, a contradiction.
  \end{proof}
\end{document}